\documentstyle[extras,11pt]{article}

\newcommand{\rr}{\mbox{\boldmath $r$}}
\newcommand{\xx}{\mbox{\boldmath $x$}}
\newcommand{\uu}{\mbox{\boldmath $u$}}

\newcommand{\pp}{\mbox{\boldmath $p$}}

\newcommand{\yy}{\mbox{\boldmath $y$}}

\newcommand{\al}{\mbox{$\alpha$}}
\newcommand{\bt}{\mbox{$\beta$}}

\newcommand{\ra}{\rightarrow}

\newcommand{\Q}{\mbox{\rm\bf Q}}
\newcommand{\Qplus}{\Q^+}

\newcommand{\maximize}{\mbox{\rm maximize }}

\newcommand{\st}{\mbox{\rm subject to }}

\newcommand{\del}{\partial}

\newcommand{\R}{\mbox{\rm\bf R}}
\newcommand{\Rplus}{\R_+}

\newcommand{\implies}{\mbox{${\Rightarrow}$}}

\makeatletter
\def\fnum@figure{{\bf Figure \thefigure}}
\def\fnum@table{{\bf Table \thetable}}

\long\def\@mycaption#1[#2]#3{\addcontentsline{\csname
 ext@#1\endcsname}{#1}{\protect\numberline{\csname
  the#1\endcsname}{\ignorespaces #2}}\par
     \begingroup
       \@parboxrestore
          \small
       \@makecaption{\csname fnum@#1\endcsname}{\ignorespaces
#3\endgroup}
      }

\begin{document}

\title{Non-Separable, Quasiconcave Utilities are Easy -- in a \\
Perfect Price Discrimination Market Model\thanks{College of Computing,
Georgia Institute of Technology, Atlanta, GA 30332--0280.
Email: {\sf vazirani@cc.gatech.edu}}
}

\author{
{\Large\em Vijay V. Vazirani}\thanks{Research supported by NSF Grants CCF-0728640 and CCF-0914732, ONR Grant N000140910755, 
and a Google Research Grant.}
}

\date{}
\maketitle

\begin{abstract}
Recent results, establishing evidence of intractability for such restrictive utility functions as additively separable,
piecewise-linear and concave, under both Fisher and Arrow-Debreu market models, have prompted
the question of whether we have failed to capture some essential elements of real markets, which seem to
do a good job of finding prices that maintain parity between supply and demand.

The main point of this paper is to show that even non-separable, quasiconcave utility functions
can be handled efficiently in a suitably chosen, though natural, realistic and useful, market model; our model allows for perfect price discrimination.
Our model supports unique equilibrium prices and, for the restriction to concave utilities, satisfies both welfare theorems.
\end{abstract}

\newpage

\section{Introduction}

The celebrated Arrow-Debreu theorem \cite{AD}, which establishes the existence of equilibria in a very general
model of the economy, has been deemed to be ``highly non-constructive'' since it crucially uses Kakutani's fixed point
theorem; as shown by Uzawa \cite{Uzawa}, the existence of general equilibrium is equivalent to fixed point 
theorems. The conditions imposed on utility functions of buyers in the Arrow-Debreu theorem
are very weak: continuity, quasiconcavity, and non-satiation. 

Over the last decade, there has been a surge of interest within theoretical computer science on studying the question of efficient
computability of market equilibria -- not only to provide an algorithmic ratification of Adam Smith's ``Invisible hand of the market'' 
but also because of potential applications to new markets on the Internet. This study started with highly restricted utility functions, 
i.e., linear \cite{DPSV,JainAD}, before moving to more general functions.
Several specialized classes of non-linear utility functions were also solved, e.g., Cobb-Douglas, CES and Leontief, see \cite{CV,Ye,CV.chap6}.

Regarding attacking general classes of utility functions, the obvious next question was additively
separable, piecewise-linear, concave utility functions (plc utilities). This long-standing open
problem was resolved recently for both market models. First, \cite{Chen.plc} proved PPAD-hardness
for the Arrow-Debreu model under plc utilities. Subsequently, PPAD-hardness was established
for Fisher's model, independently and concurrently by \cite{ChenTeng,VY.plc} and \cite{VY.plc}. Membership in PPAD
was established for both models in \cite{VY.plc}, hence precisely pinning down the complexity of plc utilities.

These results dealt a serious blow to the program of algorithmically ratifying the ``Invisible hand of the market" -- assuming 
P $\neq$ PPAD, these results effectively rule out the existence of efficient algorithms for almost all general and interesting classes 
of ``traditional'' market models.
On the other hand, markets in the West, based on Adam Smith's free market principle, seem to do a good job of finding prices
that maintain parity between supply and demand\footnote{For example, in the West, it is hard to see a sight that was commonplace in the Soviet Union,
with massive surpluses of some goods and empty shelves of others.}. This has prompted the question of whether we have failed to
capture some essential elements of real markets in our models, see \cite{va.spending}. Some progress has been made on this latter question:
polynomial time algorithms were given for spending constraint utilities \cite{va.spending} and for plc
utilities in the Fisher model, provided perfect price discrimination is introduced in the model \cite{GV.pd}. Both these works deal with
additively separable utility functions. 

We remark that
convex programs have turned out to be {\em the tool} \ for efficient computation of equilibria. For several utility functions, equilibria
of the corresponding markets have been captured as optimal solutions of convex programs, which have been solved by either combinatorial 
algorithms (see \cite{va.NB}) or using continuous algorithms, e.g., the ellipsoid method \cite{GLS}.
To our knowledge, so far all markets that have 
yielded to efficient algorithms, either exact or approximable to any degree of accuracy, have done so via this tool.

Clearly, for traditional market models, the gap between the ``positive'' algorithmic results summarized above and the generality
of the Arrow-Debreu Theorem is rather large.
The main point of this paper is to show that even non-separable, quasiconcave utility functions, with the additional restrictions of
continuous differentiability and non-satiation, can be handled efficiently in a suitably chosen, though
natural, realistic and useful, market model; our model allows for perfect price discrimination.
Our work also provides insights into the widely used practice of price discrimination, and in Section \ref{sec.appl}, we give an application 
of our market model to online display advertising marketplaces.

\subsection{Price discrimination and our results}

Most businesses today charge different prices from
different consumers for essentially the same goods or
services in order to maximize their revenues. This practice,
called {\em price discrimination}, is not only good for businesses but also
customers -- without it, some customers will simply not be able to avail of certain goods or services. 
It is not only widespread \cite{Varian.pd} but is also essential for the survival of certain
businesses, e.g., in the airline industry. 

Price discrimination is particularly important in new industries, such as telecommunications and information
services and digital goods. Traditional economic analysis, which assumes decreasing returns to scale on production,
recommends pricing goods at marginal cost. However, this is not relevant to the new industries, 
since they have very high fixed costs and low marginal costs, and hence such prices
will not even recover the fixed costs. In these situations, product differentiation and price discrimination are
an important recourse. Motivated by these considerations, price discrimination has been extensively studied in economics
from many different angles; see \cite{MacLeod.pd,Varian.pd,varian.price,arun.price,Edlin,Edlin2,bhaskar.pd} 
for just a small sampling of papers on this topic.

A monopolistic situation in which the business separates the market into individual consumers and charges
each one prices that they are {\em willing and able to pay} is called {\em perfect price discrimination}, sometimes
also called first degree price discrimination \cite{varian.price}.
More formally, a consumer's marginal willing to pay is made equal to the marginal cost of the good.
Of course, to do this, the business needs to have complete information about each consumer's preferences. 

For the restriction to concave utilities, we give a convex program, 
a generalization of the classic Eisenberg-Gale convex program, that captures equilibrium for this model.
For this case, we prove both welfare theorems. 

For quasiconcave utilities, we give a nonlinear program that captures equilibria. Similar to the convex program mentioned 
above, an optimal solution to this program also satisfies KKT conditions; moreover, this program also lends itself to a polynomial time 
solution using the ellipsoid algorithm. For this case, the first welfare theorem holds but the second welfare theorem fails;
the latter fact is established in Section \ref{sec.ex}.

\section{The Market Model}
\label{sec.model}

Our market model is based on the Fisher setting and consists of a seller with a set $G$ of divisible goods, a set $B$ of buyers each with
money and a middleman. Assume that $|G| = g$ and $|B| = n$, and the goods are numbered from 1 to $g$ and the buyers 
are numbered from 1 to $n$. Let $m_i \in \Qplus$ dollars be the money of buyer $i$. For each buyer $i$ we are specified a function
$f_i: \Rplus^g \ra \Rplus$ which gives the utility derived by $i$ as a function of allocation of goods. We will assume that $f_i$ is
polynomial time computable.

The middleman buys goods from the seller, who charges the middleman in  the usual manner, i.e., depending on the prices of goods 
and the amounts bought. However, in selling goods, the middleman charges buyers on the
basis of the utility they accrue rather than the amount of goods they receive, i.e., he price-discriminates.
The rate $r_i$ at which buyer $i$ should get utility per dollar charged from her, at any given prices $\pp$,
{\em is determined by buyer $i$ herself}. Each buyer has no utility for money but wants to maximize the utility she accrues.
The only restriction is that the middleman refuses to sell any part of a good at a loss -- the fact that the middleman knows buyers'
utility functions enables him to do this (we will specify in Section \ref{sec.rate} what this restriction means mathematically).
We show in Lemma \ref{lem.rate} that under
these circumstances, there is a rate, as a function of prices, at which buyer $i$ is able to maximize her
utility. This is also the rate at which each buyer's marginal willingness to pay equals the marginal prices of goods she gets, as
required under perfect price discrimination. At this rate $r_i$, the total utility buyer $i$ is able to get will be $r_i \cdot m_i$.

In our model, the elasticity among consumers leads to profit for the middleman; in particular, if the utility functions of all buyers are linear,
then the middleman will make no profit. 
We will study the following two cases of utility functions; clearly, the first is a special case of the second, but it has stronger properties.

\begin{itemize}
\item 
{\bf Case 1 utility functions:}
Non-separable, continuously differentiable, concave functions satisfying non-satiation.

\item 
{\bf Case 2 utility functions:}
Non-separable, continuously differentiable, quasiconcave functions satisfying non-satiation. 
\end{itemize}

We will use the following notation and definitions throughout. $\xx$ will denote allocations made of all goods to all buyers.
$\xx_i$ will denote the restriction of $\xx$ to allocations made to buyer $i$ only, and $x_{ij}$ will denote the 
amount of good $j$ allocated to buyer $i$. For the sake of ease of notation, let us introduce the following w.r.t. a generic buyer:
$\yy$, a vector of length $g$, will denote allocation and its $j$th component, $y_j$, will denote the allocation of good $j$. 
$f: \Rplus^g \ra \Rplus$ will denote her utility function.
Function $f$ is {\em concave} if for any allocations $\yy$ and $\yy^{'}$, 
\[ f\left({{\yy + \yy^{'}} \over 2}\right) \geq  {{f(\yy) + f(\yy^{'})} \over 2} .\]
Function $f$ is {\em quasiconcave} if each of its
{\em upper level sets} is convex, i.e., $\forall a \geq 0$, the set $S_a = \{ \yy \in \Rplus^g \ | \ f(\yy) \geq a \}$ is convex.
We will say that $f$ {\em satisfies non-satiation} if for any allocation $\yy$, 
there is an allocation $\yy^{'}$ that weakly dominates $\yy$ component wise and such that $f(\yy^{'}) > f(\yy)$.

The overall objective is to find prices for goods such that under these transactions, the market
clears, i.e., there is no surplus or deficiency of any good. These will be called {\em equilibrium prices}.
More formally, let $\pp$ be prices of goods and $\rr$ be the corresponding rates of buyers, as given by Lemma \ref{lem.rate}.
Assume that each buyer $i$ is charged at rate $r_i$ and is allocated a bundle of goods. We will say that prices $\pp$
are equilibrium prices if they satisfy the following {\bf conditions:}
  
\begin{enumerate}
\item 
Each good having positive price is completely sold.

\item 
The money spent by each buyer $i$ equals $m_i$.

\item
The middleman never allocates any portion of a good at a loss (the implication of this condition on allocations is given in 
Section \ref{sec.rate}).
\end{enumerate}

\subsection{Applying non-separable utilities to online display advertising marketplaces}
\label{sec.appl}

\cite{GV.pd} show how the perfect price discrimination market model can be applied to online display advertising marketplaces.
Their model assumes additively separable utilities. We show that extending the model from separable to non-separable 
utilities makes it even more relevant to this application. 

We first recall the setting from \cite{GV.pd}, which applies to companies that sell ad slots on web sites to advertisers. 
\cite{GV.pd} view such a company as the middleman, the owners of web sites as sellers and the advertisers as buyers.
They view ad slots on different web sites as goods which need to be priced. An advertiser's utility for
a particular ad slot is determined by the probability that her ad will get clicked if it is shown on that slot.
Advertisers typically pay at a fixed rate to the middleman
for the expected number of clicks they get, i.e., they are paying at fixed rate for every unit of utility they get.
Using knowledge of the utility function of buyers, the middleman is able to price discriminate.

In the model of \cite{GV.pd}, an advertiser's total utility is additive over all the slots she is allocated. 
We note that the utility to an advertiser from placing ads on multiple 
web sites would typically be an involved, non-separable function because the web sites may be substitutes, complements, etc.
Hence, extending to non-separable utilities makes the model more relevant to this application.

\subsection{Determining buyers' rates}
\label{sec.rate}

W.r.t. any prices, we will give a closed-form definition of each buyer $i$'s rate, $r_i$; for ease of notation, we will do this
for the generic buyer, i.e., we will define her rate $r^*$. For this section, assume that prices of goods are set to $\pp$.
In Lemma \ref{lem.rate} we will show $r^*$ is indeed her optimal rate, i.e., it maximizes her utility. 
In Section \ref{sec.program} we will show that the solution of the convex (nonlinear) program will
assign utilities to a buyer at precisely this rate w.r.t. equilibrium prices. Hence, there is no need to explicitly compute buyers' rates. 

Given two allocations $\yy$ and $\yy'$, we will say that $\yy$ {\em weakly dominates} $\yy'$ if for each good $j$,
$y_j \geq y_j'$. The next lemma gives the mathematical condition that an allocation needs to satisfy in order to satisfy Condition 3.

\begin{lemma}
\label{lem.condition}
An allocation $\yy$ made by the middleman at rate $r$ satisfies Condition 3 iff
\[ \forall \yy' \ s.t. \ \yy \ \mbox{weakly dominates} \ \yy' , \forall j :  \ \  \left( {{\del f} \over {\del y_j}} (\yy')  \ \div \ r  \right)  
\  \geq  \  p_j  .\]
\end{lemma}

\begin{proof}
Conceptually, assume that the middleman is making an allocation to the buyer gradually and continuously and
is charging the buyer at the rate of $r$ units of utility per dollar. 
Clearly, the effective price at which he is selling her good $j$ depends on the allocation made already. If the latter is $\yy$,
then the {\em marginal price of good j at allocation} $\yy$ is
\[ {{\del f} \over {\del y_j}} (\yy)  \ \div \ r .\] 
Therefore, at this point the middleman is selling good $j$ at a loss iff 
\[ \left( {{\del f} \over {\del y_j}} (\yy)  \ \div \ r  \right)  \  <  \  p_j  ,\] 
since then he is charging the buyer less for good $j$ than the amount charged from him by the seller. 
The lemma follows.
\end{proof}

Let us say that an allocation $\yy$ is {\em feasible for rate r} if it satisfies the condition given in Lemma \ref{lem.condition}.
We next define the set of maximal, under the relation ``weakly dominates'', allocations that are feasible.
For $r > 0$, this set is:
\[ S(r) \ = \ \left\{ \yy \ | \ \left( {{\del f} \over {\del y_j}} (\yy)  \ \div \ p_j \right) \ = \ r \ 
\ \mbox{if} \ \ y_j > 0 \ \ \mbox{and} \ \leq r \ \mbox{otherwise} \right\} . \]
Observe that if $f$ is strictly concave, $S(r)$ will be a singleton for each $r$. 
The function $U: \Rplus \ra \Rplus$ gives the largest utility attained by a feasible allocation at rate $r$:
\[ U(r) = \sup \{ f(\yy)  \ | \ \yy \in S(r) \} . \]
Clearly, $U$ is a decreasing function of $r$.
Observe that because of the non-satiation condition, $\lim_{r \rightarrow 0} U(r)$ is unbounded.

Finally, we define rate $r^*$ as follows
\[ r^*  \   =  \  \arg \max_{r} \{ U(r) \geq m \cdot r \} , \]
where $m$ is the money of the generic buyer. 
Since function $U(r)$ is unbounded as $r \rightarrow 0$, $r^*$ is well defined for all $m$.

\begin{lemma}
\label{lem.rate}
$r^*$ maximizes the utility accrued by the generic buyer.
\end{lemma}

\begin{proof}
Since $U(r)$ gives the maximum utility that the buyer can get from a feasible allocation at rate $r$, and since
this is a decreasing function of $r$, if the rate is fixed at $\bt > r^*$, then
the utility accrued by the buyer is $U(\bt) \leq U(r^*)$.
The maximum amount of utility the buyer can get for money $m$ at rate $r$ is $r \cdot m$. Therefore, if the rate is fixed at $\al < r_i$, then 
the utility accrued by $i$ is $\al \cdot m < r^* \cdot m$. Hence $r^*$ is the optimal rate.
\end{proof}

\section{The Convex/Nonlinear Program}
\label{sec.program}

The program (\ref{CP}) given below is a convex program for Case 1 utility functions and
simply a nonlinear program for Case 2 utility functions. Besides non-negativity, the only constraint is
that at most 1 unit of each good is sold. We will denote the Lagrange variables corresponding to these
constraints as $p_j$'s and will show that at optimality, they will be equilibrium prices of the corresponding market.

\begin{lp}
\label{CP}
\maximize & \sum_{i \in B} m_i \log (f_i(\xx_i))  \\[\lpskip]
\st       & \forall j\in G: \ \sum_{i\in B}  x_{ij} \leq 1  \nonumber \\
          & \forall i \in B, \ \forall j\in G: \ x_{ij}   \geq 0 \nonumber 
\end{lp}

The KKT conditions for this program are:

\begin{enumerate}
\item 
\[ \forall j \in G: \ \  p_j \geq 0 .\]

\item 
\[ \forall j \in G: \ \ p_j > 0  \ \  \implies \ \  \sum_{i\in B} x_{ij} = 1 .\]

\item 
\[ \forall i \in B, \ \forall j \in G: \ \ p_j \geq {m_i \over {f_i(\xx_i)}} \cdot {{\del f_i} \over {\del x_{ij}}} (\xx_i) .\]

\item 
\[ \forall i \in B, \ \forall j \in G: \ \ x_{ij} > 0 \ \ \implies \ \ p_j =
{m_i \over {f_i(\xx_i)}} \cdot {{\del f_i} \over {\del x_{ij}}} (\xx_i) . \]
\end{enumerate}

\begin{theorem}
\label{thm.eq}
For both cases of utility functions, the optimal primal and dual solutions to program (\ref{CP}) give equilibrium allocations and prices, and the 
latter are unique. Moreover, both can be computed to any required degree of accuracy in polynomial time.
\end{theorem}

\begin{proof}
Because utility functions are assumed to be continuously differentiable, for an optimal solution to program (\ref{CP}) there is a unique
dual, i.e., prices, that satisfies the KKT conditions stated above. From these, we will derive the 3 conditions defining equilibrium.
For Case 1 utility functions, (\ref{CP}) is a convex program and for Case 2 utility functions, the upper level sets are convex, and
for both cases a separation oracle can be implemented in polynomial time. Moreover, since the constraints are all linear,
by \cite{GLS}, the optimal solutions can be computed in polynomial time to any required degree of accuracy.

The first equilibrium condition is implied by the KKT conditions 1 and 2.
Consider buyer $i$. Because of non-satiation, $x_{ij} > 0$ for some $j$. For this $j$, let
\[ r_i = \left( {{\del f_i} \over {\del x_{ij}}} (\xx_i) \ \div \ p_j \right) . \]
By KKT condition 4, any good $j$ with $x_{ij} > 0$ must satisfy this equality, and
if for some good $j$, $x_{ij} = 0$, then by KKT condition 3,
\[ r_i \geq \left( {{\del f_i} \over {\del x_{ij}}} (\xx_i) \ \div \ p_j \right) . \]
This proves that the middleman does not sell any part of a good at a loss. 
Substituting $r_i$ back in KKT condition 4, we get $m_i = r_i \cdot f_i(\xx_i)$, thereby proving that all money of 
buyer $i$ is spent and $r_i$ is the  rate whose existence is established in Lemma \ref{lem.rate}.
\end{proof}

\section{The Welfare Theorems}
\label{sec.welfare}

The {\em first welfare theorem}  states that allocations made at equilibrium prices are Pareto optimal and
the {\em second welfare theorem} states that for any Pareto optimal utilities $\uu^*$, there is a way of setting
the initial moneys of buyers in such a way that an equilibrium obtained for this instance gives precisely $\uu^*$
utilities to buyers.

\begin{theorem}
\label{thm.Pareto}
The first welfare theorem is satisfied by both cases of utility functions and the second welfare theorem is
satisfied by Case 1 utility functions.
\end{theorem}

\begin{proof}
If an allocation is not Pareto optimal, then it cannot maximize the objective function of program (\ref{CP}). Hence,
the first statement follows from Theorem \ref{thm.eq}.

Next assume that the utility functions are concave. Let $S$ be the image of the map from the polytope of feasible allocations
to utilities accrued. Because of concavity of the utility functions, $S$ is a convex region in $\Rplus^n$. Let $\uu^*$ be any Pareto 
optimal utilities. Clearly, $\uu^*$ must lie on the boundary of this convex region. Let $\sum_{i \in B} {a_i \cdot u_i} = c$ be the hyperplane 
that is tangent to region $S$ and contains the point $\uu^*$. By Pareto optimality of $\uu^*$, $a_i \geq 0$, for each $i \in B$,
with at least one of these values being positive.

Let $m_i = a_i \cdot u_i^*$. Find the largest value of $c'$ so that $\sum_{i \in B} {m_i \log{u_i}} = c'$ contains the point $\uu^*$.
Now, it is easy to see that for the chosen  values of $m_i$'s, $\sum_{i \in B} {m_i \log{u_i}} = c'$ must be tangent to $S$ at $\uu^*$.
Therefore, if the market is run with the money of each buyer $i$ set to $m_i$, the equilibrium utilities will be $\uu^*$, hence proving
the second welfare theorem.
\end{proof}

\subsection{The second welfare theorem fails for Case 2 utility functions}
\label{sec.ex}

We first give an example that uses continuous, quasiconcave utility functions, with one good and 2 buyers, 
for which the second welfare theorem fails. Let $0 < a < 1/2$ be a constant.
Define the utility function of each buyer for the good to be $f(x) = 2ax$, for $0 \leq x \leq 1/2$, and 
$f(x) = 2(1 - a)x + (2a - 1)$, for $1/2 \leq x \leq 1$.
Then, the region $S$ defined in the proof of Theorem \ref{thm.Pareto} is non-convex.

The point $\uu^* = (a, a)$ lies on the boundary of $S$ and hence constitutes Pareto optimal utilities.
Consider the hyperbola $u_1 u_2 = a^2$, which contains the point $\uu^*$.
It is easy to see, especially for small values of $a$, that this hyperbola is not tangent to region $S$. Indeed, because
of the non-convexity of $S$, for any value of $m_1, m_2$, the hyperbola $u_1^{m_1} u_2^{m_2} = a^{m_1 + m_2}$, which contains the point $\uu^*$,
is not tangent to $S$. Therefore, $\uu^*$ can never be equilibrium utilities of this market, for any value of $m_1, m_2$.
In particular, when run with money
$m_1 = m_2$, equilibrium utilities are not attained at $\uu^*$ because the objective function of convex program (\ref{CP}) is not maximized 
at this point (observe that despite the symmetry in this case, equilibrium utilities will be attained at a point $u_1 \neq u_2$).
Hence, the second welfare theorem fails for this example. Finally, by smoothening out, $f(x)$ can be made continuously differentiable at $x = 1/2$ 
without substantially changing $S$ or the conclusion derived above.

\bibliography{kelly}
\bibliographystyle{alpha}

\end{document}